\documentclass[lettersize,journal]{IEEEtran}

\usepackage{bm}
\usepackage{makeidx}
\usepackage{enumerate}
\usepackage{color}
\usepackage{cite}
\usepackage{amsmath,amsthm}   
\usepackage{amssymb}
\usepackage{multirow}
\usepackage{graphicx}
\usepackage{epstopdf}
\usepackage{balance}
\usepackage{multicol}
\DeclareGraphicsExtensions{.pdf,.jpeg,.png,.jpg,.emf,.eps}
\usepackage{calc}
\usepackage{url}

\usepackage{upref}
\usepackage{comment}
\usepackage{times}
\usepackage{dsfont}
\usepackage{rawfonts}
\usepackage[T1]{fontenc}
\usepackage{latexsym}
\usepackage{amsfonts}

\usepackage[font=small,labelfont=bf]{caption}
\usepackage[font=small,labelfont=bf]{subcaption}
\usepackage{optidef}
\usepackage[compress]{cleveref}

\usepackage{xcolor}
\usepackage{tikz,pgfplots}
\usetikzlibrary{calc}
\usetikzlibrary{intersections}
\usetikzlibrary{arrows,shapes}
\usetikzlibrary{positioning}
\usepackage[ruled,noline,linesnumbered]{algorithm2e}

\newtheorem{proposition}{Proposition}

\newtheorem{remark}{Remark}

\theoremstyle{definition}

\def\Q{{\bf Q}}
\def\I{{\bf I}}
\def\A{{\bf A}}
\def\U{{\bf U}}
\def\V{{\bf V}}
\def\F{{\bf F}}
\def\T{{\bf T}}
\def\G{{\bf G}}
\def\H{{\bf H}}
\def\U{{\bf U}}
\def\B{{\bf B}}
\def\C{{\bf C}}
\def\E{{\bf E}}

\def\f{{\bf f}}
\def\g{{\bf g}}

\def\Thetab{\bm{\Theta}}

\def\Lambdab{\bm{\Lambda}}

\def\det{\operatorname{det}}
\def\diag{\operatorname{diag}}
\def\Re{\operatorname{Re}}
\def\rank{\operatorname{rank}}
\def\Exp{\operatorname{E}}
\def\blkdiag{\operatorname{blkdiag}}

\begin{document}

\title{Rate Analysis and Optimization of LoS Beyond Diagonal RIS-assisted MIMO Systems}

\author{Ignacio Santamaria, \IEEEmembership{Senior Member, IEEE}, Jes{\'u}s Guti{\'e}rrez, \IEEEmembership{Member, IEEE}, Mohammad Soleymani, \IEEEmembership{Member, IEEE},
Eduard Jorswieck, \IEEEmembership{Fellow, IEEE}

\thanks{I.~Santamaria is with the Department of Communications Engineering, Universidad de Cantabria, Santander, Spain (e-mail: i.santamaria@unican.es). }
\thanks{M. Soleymani is with the Signal and System Theory Group, Universit{\"a}t  Paderborn, 33098 Paderborn, Germany (e-mail : mohammad.soleymani@sst.upb.de).}

\thanks{E. Jorswieck is with the Institute for Communications Technology, Technische Universität Braunschweig, 38106 Braunschweig,
Germany (e-mail: e.jorswieck@tu-bs.de).}
\thanks{J. Guti{\'e}rrez is with IHP - Leibniz-Institut
f{\"u}r Innovative Mikroelektronik, 15236 Frankfurt (Oder), Germany (e-mail: teran@ihp-microelectronics.com).}

\thanks{This work is supported by the European Commission’s Horizon Europe, Smart Networks and Services Joint Undertaking, research and innovation program under grant agreement 101139282, 6G-SENSES project. The work of I. Santamaria was also partly supported under grant PID2022-137099NB-C43 (MADDIE) funded by MICIU/AEI /10.13039/501100011033 and FEDER, UE.}}

\markboth{IEEE Communication Letters}{}

\maketitle

\begin{abstract}

In this letter, we derive an expression for the achievable rate in a multiple-input multiple-output (MIMO) system assisted by a beyond-diagonal reconfigurable intelligent surface (BD-RIS) when the channels to and from the BD-RIS are line-of-sight (LoS) while the direct link is non-line-of-sight (NLoS). The rate expression allows to derive the optimal unitary and symmetric scattering BD-RIS matrix in closed form. Our simulation results show that the proposed solution is competitive even under the more usual Ricean channel fading model when the direct link is weak.

\end{abstract}


\IEEEpeerreviewmaketitle

\section{Introduction}
Beyond diagonal RISs (BD-RISs) are being intensively studied lately as they provide enhanced control over the amplitude and phase of the reflecting elements and thus greater flexibility than diagonal RIS \cite{ClerckxTWC22a,ClerckxTWC22b,MaoCL2024,Wymeersch2024Arxiv,khan2025surveydiagonalrisenabled}. The scenarios where RISs provide the most significant gains are those where the direct channel is modeled as weak non-line-of-sight (NLoS) between the transmitter (Tx) and the receiver (Rx) (or may even be obstructed), and Ricean with a predominant line-of-sight (LoS) path \cite{Pei2021LoS} for the channels through the RIS. Several BD-RIS-assisted scenarios have been studied assuming this channel model, but the symmetry and unitarity constraints on the passive BD-RIS scattering matrix usually lead to iterative algorithms with high computational complexity \cite{ClerckxTWC24a, ClerckxTWC22b, SoleymaniTWC2023}.

In this letter, we consider a limiting case of this model in which the channels to and from the BD-RIS are pure LoS and address the problem of maximizing the achievable rate. The first work that studied the capacity in a MIMO link assisted by a diagonal RIS is \cite{ZhangCapacityJSAC2020}, where an alternating optimization (AO) algorithm between the transmit covariance matrix and the phase shifts of the diagonal RIS is proposed. The problem of maximizing the rate of a MIMO link assisted by a BD-RIS is addressed in \cite{SantamariaSPAWC24}, where an AO algorithm is proposed. The BD-RIS optimization solves a sequence of quadratic problems in the manifold of unitary matrices, thus the algorithm in \cite{SantamariaSPAWC24} has a high computational complexity. For single-stream transmission, given the Tx/Rx beamformers the problem of finding the BD-RIS that maximizes the achievable rate is equivalent to maximizing the signal-to-noise ratio (SNR), a problem for which a closed-form solution exists for the BD-RIS \cite{NeriniTWC2023, SantamariaSPLetters2023}. The Tx/Rx beamformers can be optimized through alternate optimization as proposed in \cite{NeriniTWC2023}. More recently, a closed-form solution to the capacity maximization problem in BD-RIS-assisted MIMO systems has been presented in \cite{Emil2024Arxiv}. Interestingly, LoS channels are also common in RIS/BD-RIS assisted sensing scenarios \cite{ShaoJSAC22}. Although space limitations preclude a more detailed review of this and other lines of recent research on BD-RIS, the reader can find comprehensive surveys in \cite{khan2025surveydiagonalrisenabled, Maraqa2025surveyBDRIS}. 

To the best of the authors' knowledge, the rate maximization problem under Tx-RIS-Rx LoS channels has not been analyzed previously. This paper derives a closed-form expression for the rate that can be maximized to obtain the optimal BD-RIS when the channels through the BD-RIS are LoS. Our simulations show that this solution is competitive even with the more usual Ricean channels with a dominant LoS path.

\textit{Notation}: A bold-faced upper case letter, $\A$, is a matrix, a bold-faced lower case letter, ${\bf a}$, is a vector, and a light-faced lower case letter, $a$, is an scalar. $\A^T$, $\A^*$, $\A^H$, $\A^{-1}$, $\det(\A)$ are, respectively, transpose, conjugate, Hermitian, inverse and determinant. $|\cdot|$, $\|\cdot \|$, and $\| \cdot\|_1$ denote the absolute value, the Euclidean $l_2$-norm, and the $l_1$-norm, respectively. $\I_n$ denotes the identity matrix of size $n$, but when there is no need the subindex will be omitted. $ {\cal CN}({\bf 0}, {\bf R})$ is the proper complex Gaussian distribution with zero mean and covariance matrix ${\bf R}$. $\Exp[\cdot]$ denotes mathematical expectation. We use $\angle a$ to denote the angle of the complex number $a$. Finally, $\odot$ denotes Hadamard product. 
%%%%%%%%%%%%%%%%%%%%%%%%%%
\section{System model}
\label{sec:model}
We consider a BD-RIS-assisted multiple-input multiple-output (MIMO) link as depicted in Fig. \ref{fig:DeploymentBDRIS}. The Tx is equipped with $N_T$ antennas, the Rx is equipped with $N_R$ antennas, and the BD-RIS has $M$ elements. The equivalent MIMO channel is
\begin{equation}
\H = \H_d + \F \Thetab \G^H,
\label{eq:MIMOchanneleq}
\end{equation}
where $\G \in \mathbb{C}^{N_T \times M}$ is the channel from the Tx to the BD-RIS, $\F \in \mathbb{C}^{N_R \times M}$ is the channel from the BD-RIS to the Rx, $\H_d \in \mathbb{C}^{N_R \times N_T}$ is the MIMO direct link, and $\bm{\Theta}$ is the $M \times M $ BD-RIS matrix. In this paper, $\G$ will sometimes be referred to as the forward channel and $\F$ as the backward channel. We consider the following feasibility set for the fully-connected BD-RIS \cite{ClerckxTWC22a}, \cite{ClerckxTWC22b}
\begin{equation*}
{\cal{T}} = \{ \Thetab = \Thetab^T, \Thetab^H \Thetab = \I_M \}. 
\end{equation*}

The BD-RIS is deployed at sufficient height to ensure a direct LoS from the Tx and to the Rx. Therefore, the forward and backward BD-RIS channels are assumed to be pure LoS channels. However, the direct MIMO channel, $\H_d$, is assumed to be an NLoS channel with a rich multipath scattering. LoS models are common in RIS-assisted scenarios \cite{wu2019towards,ZhangRankOneWCL2021,OzdoganWCL2020,EsmaeilbeigSPL2025,Utschik2024Arxiv}. In the mmWave band, the LoS component is dominant as the path loss of the NLoS channel is usually much higher \cite{AkdenizJSAC14}, thus making the LoS approximation reasonable in such bands. Another situation that gives rise to LoS channels is when the RIS is close to the Tx or Rx. This scenario is considered in \cite{ZhangRankOneWCL2021}, where a channel estimation algorithm exploiting the rank-1 structure of $\G$ is proposed. A limiting case of this situation is that of transmissive RISs \cite{ZengCL2021}, in which the RIS is deployed very close to the Tx and the transmitted signal can penetrate the RIS and serve users on the opposite side. Finally, it should be noted that scenarios in which the NLoS direct Tx-Rx channel, $\H_d$, is weak and the forward and backward RIS channels, $\G$ and $\F$, are LoS, are those scenarios in which RISs provide more significant gains, as experimentally shown in \cite{Pei2021LoS}.
\begin{figure}
    \centering
\includegraphics[width=0.5\textwidth]{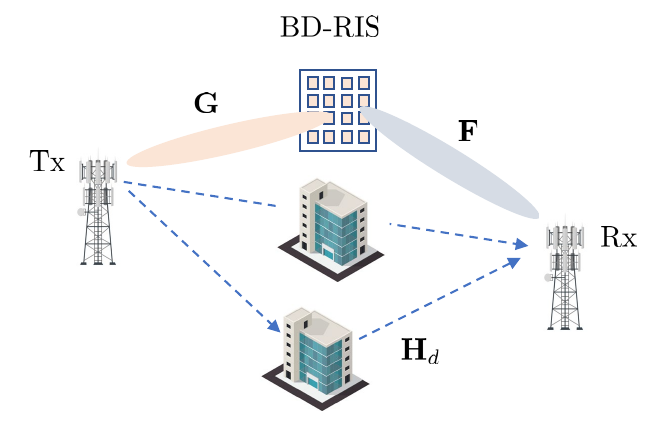}
     \caption{BD-RIS-aided MIMO communication system. The BD-RIS is strategically deployed to have a direct LoS path to the transmitter ($\G$) and receiver ($\F$). The signal from the Tx arrives at the Rx by several NLoS paths shown in dashed line.}
	\label{fig:DeploymentBDRIS}
\end{figure}

Under these assumptions $\H_d$ is a full-rank matrix in \eqref{eq:MIMOchanneleq}, while $\G$ and $\F$ are rank-1 matrices
\begin{equation}
\F = \f_{a} \f_{d}^H \quad {\rm{and}} \quad \G^H  = \g_a \g_d^H. \label{eq:pureLoS}
\end{equation}
Note that $\f_d$ (similarly for $\f_a$, $\g_a$, and $\g_d$) can be modeled as $\f_d = \beta_d \bar{\f}_d$, where the real-valued $\beta_d$ represents the propagation path loss while $\bar{\f}_d$ ($\|\bar{\f}_d\| =1$) captures local variations along the array. Typically, in LoS environments the channel vector $\bar{\f}_d$ is uniquely determined by the array geometry and the angle-of-departure. We consider a half-wavelength uniform linear array (ULA)  
\begin{equation}
\bar{\f}_d(\phi) = \frac{1}{\sqrt{M}} \left( 1, e^{-j\pi \sin(\phi)}, \ldots, e^{-j\pi \sin(\phi) (M-1) } \right)^T,
\label{eq:ULA}
\end{equation}
where $\phi$ is the angle of departure from the BD-RIS to the Rx. The equivalent channel with LoS through RIS components reduces to
\begin{equation}
\H = \H_d +  \left( \f_d^H \Thetab \g_a \right) \f_a \g_d^H = \H_d + \alpha e^{j \theta} \f_a \g_d^H,
\label{eq:Heq_rankone}
\end{equation}
where we have defined $\alpha = \f_d^H \tilde{\Thetab} \g_a$ and $\Thetab = e^{j\theta}\tilde{\Thetab}$. Therefore, the equivalent MIMO channel comprises the direct full-rank MIMO channel plus a rank-1 component scaled by a complex value with modulus $\alpha \geq 0$ and phase $\theta$, which depends on the BD-RIS.  

%%%%%%%%%%%%%%%%%%%%%%%%%%%%%%%%%%%%%%%%%%%%%%%%%%%%%%%%%
\section{Optimal BD-RIS and achievable rate}
\label{sec:cap}
In this section, we present the two main results of this work. First, Proposition \ref{prop1new} derives an expression for $\det(\I + \B \B^H)$ as a function of $\det(\I + \A \A^H)$ when $\B = \A +\alpha e^{j \theta} \f \g^H$  is a rank-1 perturbation of $\A$. This algebraic result is then used in Proposition \ref{prop:optCap}  to find the optimal BD-RIS that maximizes the rate in the MIMO scenario shown in Fig. 1.

\begin{proposition}
\label{prop1new}
Let $\B = \A + \alpha e^{j \theta} \f \g^H$ be a rank-1 perturbation of the $n \times m$ complex matrix $\A$. Then,
\begin{equation}
\hspace{-0.1cm}\det(\I + \B\B^H) = \det(\I + \A\A^H)  \left( 1 + \Delta\right), \label{eq:detexp}
\end{equation}
where  $\Delta = Z \alpha^2 + 2\alpha \Re \left ( e^{j \theta}  \gamma_3 \right)$, $Z = |\gamma_3|^2 + \gamma_1 (\|\g\|^2 - \gamma_2)$, and
\begin{subequations}
\begin{align}
    \gamma_1 &= \f^H\left( \I +\A \A^H \right)^{-1} \f, \label{eq:gamma1bis} \\
    \gamma_2 &= \g^H \A^H\left( \I + \A \A^H \right)^{-1} \A \g, \label{eq:gamma2bis} \\
    \gamma_3 &= \g^H \A^H\left( \I + \A \A^H \right)^{-1} \f. \label{eq:gamma3bis}
\end{align}
\end{subequations}
Note that $ \gamma_1 \geq 0$ and $\gamma_2 \geq 0$ are real values, whereas $\gamma_3$ is a complex scalar. 
\end{proposition}
\begin{proof}
 See Appendix A. 
 \end{proof}

Now we apply Proposition \ref{prop1new} to the problem of maximizing the transmission rate in a MIMO channel assisted by a BD-RIS. The Tx sends proper Gaussian signals ${\bf x} \sim \mathcal{CN}({\bf 0}, {\bf R}_{xx})$, where ${\bf R}_{xx} = \Exp[{\bf x}{\bf x}^H]$ denotes the Tx covariance matrix. For a fixed ${\bf R}_{xx}$, the rate maximization problem for a BD-RIS-assisted MIMO link can be formulated as follows 
\begin{equation}
\label{eq:ProbMaxCap}
 {\cal P}_1: \,\max_{\Thetab \in {\cal{T}} }\,\,  \log \det \left( \I +  \frac{1}{\sigma^2}\H {\bf R}_{xx} \H^H \right),
\end{equation}
where $\H$ is the equivalent channel given by \eqref{eq:Heq_rankone} that depends on $\Thetab$, and $\sigma^2$ is the noise variance. The optimal solution of ${\cal P}_1$ is presented in the following proposition.

\begin{proposition}
\label{prop:optCap}
Let us define the rank-2 matrix ${\bf T} = \f_d \g_a^H + \left(\f_d \g_a^H \right)^T$ and compute its singular value decomposition (SVD) as $\T = \U \Lambdab \V^H$. Let us partition the eigenspaces as $\U = [\U_1, \U_2]$ and $\V = [\V_1, \V_2]$, where $\U_1$ (resp.$\V_1$) contains the first two columns of $\U$ (resp. $\V$) corresponding to the signal subspace, and $\U_2$ (resp.$\V_2$) contains the remaining $M-2$ columns corresponding to the null subspace. The optimal solution of ${\cal P}_1$ in \eqref{eq:ProbMaxCap} is
\begin{equation}
    \Thetab_{opt} = e^{j\theta_{opt}} \left( \U_1\V_1^H + \V_2^* \Q_{rot}^* \Q_{rot}^H \V_2^H\right),
    \label{eq:optMIMOcap}
\end{equation}
where 
\begin{equation}
\theta_{opt} = -\angle\ \g_d^H {\bf R}_{xx} \H_d^H \left( \I + \frac{1}{\sigma^2}\H_d {\bf R}_{xx} \H_d^H \right)^{-1}\f_a,
\label{eq:optphase}
\end{equation}
and $\Q_{rot}$ is an arbitrary $(M-2) \times (M-2)$ unitary matrix.
\end{proposition}

\begin{proof}
Let us define $\B= \H {\bf R}_{xx}^{1/2}/\sigma$, $\A = \H_d {\bf R}_{xx}^{1/2}/\sigma$, $\g = {\bf R}_{xx}^{H/2} \g_d/\sigma$, and $\f = \f_a$. With these definitions, the Tx covariance matrix is absorbed in the equivalent channel and we have $\B = \A + \alpha e^{j \theta} \f \g^H$ with $\f_d^H \Thetab \g_a = \alpha e^{j\theta}$. Therefore, a direct application of Proposition \ref{prop1new} yields the following rate expression for a BD-RIS-assisted MIMO link
\begin{equation*}
    \det \left( \I +  \frac{1}{\sigma^2}\H {\bf R}_{xx} \H^H \right) =  \det \left( \I +  \A \A^H \right)(1+\Delta),
\end{equation*}
with $\Delta$ defined as in Proposition \ref{prop1new}. In the previous expression $\log \det \left( \I +  \A \A^H \right)$ is the rate without BD-RIS and $\log(1+\Delta)$ is the rate gain provided by the BD-RIS. Therefore, the optimal BD-RIS $\f_d^H \Thetab \g_a = \alpha e^{j\theta}$ solves the following problem
\begin{equation}
\label{eq:ProbMaxCap1}
\,\max_{\Thetab \in {\cal{T}} }\,\,  \Delta = \max_{\alpha, \theta }\,\, Z \alpha^2\,+ 2 \, \alpha \Re \left ( e^{j \theta}  \gamma_3 \right),
\end{equation}
where, according to Proposition \ref{prop1new}, $Z =|\gamma_3|^2 +\gamma_1 (\|\g\|^2 - \gamma_2)$. Next, we prove that $Z >0$. Both $|\gamma_3|^2$ and $\gamma_1$, given by \eqref{eq:gamma1bis}, are positive values, so we only have to show that $\gamma_2 \leq \|\g\|^2$. Let $\A = \U_a \Lambdab_a \V_a^H$ be the SVD of $\A$, where $\Lambdab_a = \diag(\lambda_1,\ldots,\lambda_n)$ and $n = \rank(\A)$. From \eqref{eq:gamma2bis}, we can write
\begin{equation*}
 \gamma_2 =\g^H \V_a \diag\left(\frac{\lambda_1^2}{1+\lambda_1^2}\ldots, \frac{\lambda_{n}^2}{1+\lambda_{n}^2}\right) \V_a^H \g,
\end{equation*}
where all values in the diagonal matrix are smaller than one, which implies $\gamma_2 \leq \|\g\|^2$ thus proving that $Z > 0$. Note that this result is valid for any Tx covariance matrix ${\bf R}_{xx}$.

Since $Z$ is a positive value, the optimal phase that maximizes \eqref{eq:ProbMaxCap1} is 
\begin{equation}
    \theta_{opt} = -\angle\,\gamma_3 =  -\angle \g^H \A^H\left( \I + \A \A^H \right)^{-1} \f.
    \label{eq:phaseopt}
\end{equation}
Substituting in \eqref{eq:phaseopt} $\A = \H_d {\bf R}_{xx}^{1/2}/\sigma$, $\g = {\bf R}_{xx}^{H/2} \g_d/\sigma$, $\f = \f_a$, and scaling the resulting expression by $\sigma^2$, which does not change the phase, we get the optimal phase in \eqref{eq:optphase}. Therefore, maximizing the achievable rate amounts to maximizing $\alpha$. The maximum value is $\alpha_{BD-RIS} = \|\f_d\| \|\g_a\|$, which is achieved by a BD-RIS obtained from the Takagi's factorization of ${\bf T} = \f_a \g_d^H + (\f_a \g_d^H)^T$ as proved in \cite{SantamariaSPLetters2023}. Equivalently, the solution can be written using the SVD of $\T$ as described in Proposition \ref{prop:optCap} (see also \cite{MaoCL2024}). This completes the proof.
\end{proof}
Proposition 2 can be applied in an alternating optimization procedure summarized in Algorithm 1 that optimizes ${\bf R}_{xx}$ via waterfilling over the equivalent channel $\H_{eq}$ for a given BD-RIS, and then optimizes the BD-RIS for the new ${\bf R}_{xx}$. The procedure is initialized with an isotropic matrix ${\bf R}_{xx} = \frac{P_t}{N_T} \I_{N_T}$ where $P_t$ denotes the Tx power.

In \eqref{eq:optMIMOcap}, the solution with $\Q_{rot} = {\bf 0}$ leads to a rank-2 lossy BD-RIS matrix ($\Thetab_{opt}^H \Thetab_{opt} \prec \I$) that optimally reflects the incident signal from the forward direction $\g_a$ to the backward direction $\f_d$. That is, optimal BD-RIS performance can be achieved in this scenario with a reflected power lower than the incident power. It follows from Proposition \ref{prop:optCap} that the rate improvement of an optimal BD-RIS as compared to a scenario without RIS is
\begin{equation}
\hspace{-0.2cm}\log(1 +\Delta) = \log \left(1 + \|\f_d\|^2 \|\g_a\|^2 Z + 2 \|\f_d\| \|\g_a\| |\gamma_3| \right).
   \label{eq:DeltaBDRIS}
\end{equation}

\begin{algorithm}[!t]
\small
\DontPrintSemicolon
\SetAlgoVlined
\KwIn{Initial $\H, \F = \f_{a} \f_{d}^H, \G  = \g_d \g_a^H$, and  ${\bf R}_{xx} = \frac{P_t}{N_T} \I_{N_T}$, noise variance estimate $\sigma^2$}
\KwOut{Final BD-RIS $\Thetab$ and ${\bf R}_{xx}$}
{Compute $\T = \f_d \g_a^H + \left(\f_d \g_a^H \right)^T$} \;
{Compute $\T = \U \Lambdab \V^H$ and partition $\U = [\U_1, \U_2]$ and $\V = [\V_1, \V_2]$ as in Proposition \ref{prop:optCap} }\;
{Generate a random unitary $\Q_{rot} \in {{\cal{U}}(M-2)}$ }\;
{Obtain  $\tilde{\Thetab} = \left( \U_1\V_1^H + \V_2^* \Q_{rot}^* \Q_{rot}^H \V_2^H\right)$}\;
\While{Convergence criterion not true}
{
%Compute $\A = \H_d {\bf R}_{xx}^{1/2}/\sigma$, $\g = {\bf R}_{xx}^{H/2} \g_d/\sigma$, $\f = \f_a$ \;
Estimate $\theta_{opt}$ as \eqref{eq:optphase} and $\Thetab = e^{j\theta_{opt}} \tilde{\Thetab}$ \;
Obtain ${\bf R}_{xx}$ via waterfilling over $\H_d + \F \Thetab \G^H$
}
\caption{{\small Proposed BD-RIS with optimal ${\bf R}_{xx}$}}
\label{alg:BDRIS}
\end{algorithm}
  
\begin{remark}[{\bf Diagonal RIS}] The optimal solution for a diagonal RIS can also be expressed as $\Thetab = e^{j\theta_{opt}}\tilde{\Thetab}$, with the optimal phase $\theta_{opt}$ given by \eqref{eq:phaseopt}. The maximum achievable $\alpha$ for an RIS is $\alpha_{RIS} = \| \f_d^* \odot \g_a  \|_1$, which is attained by $\tilde{\Thetab} = \diag(e^{j \tilde{\theta}_1},\ldots, e^{j \tilde{\theta}_M})$ with phases $\tilde{\theta}_m = -\angle \f_d(m)^*\g_a(m)$, $m=1,\ldots, M$. 
\end{remark}
In the case of pure $\f_d$ and $\g_a$ LoS channels defined as in \eqref{eq:ULA}, $\| \f_d^* \odot \g_a  \|_1 = \|\f_d\| \|\g_a\|$ ($\alpha_{RIS} = \alpha_{BD-RIS}$) and, therefore, a diagonal RIS achieves the same rate gain as a BD-RIS, thus corroborating the results in \cite{ClerckxTWC22a}. Even so, deploying a BD-RIS can bring advantages that make it an interesting alternative. For example, using $\Q_{rot} = {\bf 0}$, a BD-RIS achieves greater directionality in the reflected signal thus reducing interference in unwanted directions. Furthermore, as we will see in the simulations, the BD-RIS solution performs better than the conventional diagonal RIS in the more realistic scenario where forward and backward channels are Rician. 

\begin{remark}[{\bf Group-connected BD-RIS}] To reduce the computational and circuit complexity of the fully-connected BD-RIS architecture, a group-connected architecture was proposed in \cite{ClerckxTWC22a, ClerckxTWC22b}. The $M$ reflective elements of a group-connected BD-RIS are partitioned into $G$ groups, each of $M_g = M/G$ elements. The elements of each group are fully connected but disconnected from the other groups. The result for the fully-connected BD-RIS extends directly to the group-connected architecture for which the scattering matrix is a block-diagonal matrix $\Thetab = \blkdiag(\Thetab_1,\ldots, \Thetab_G)$. For a group-connected BD-RIS the equivalent channel in \eqref{eq:Heq_rankone} can be expressed as
\begin{equation}
    \H = \H_d + e^{j \theta} \left (\sum_{g=1}^G \alpha_g e^{j \theta_g} \right) \f_a \g_d^H,
\end{equation}
where $\theta$ is a common phase term for the BD-RIS elements and $\alpha_g e^{j \theta_g} = \f_{d,g}^H \Thetab_g \g_{a,g}$ represents the response of the $g$th group. The group-connected BD-RIS that maximizes the achievable rate has a common phase $\theta$ given by \eqref{eq:optphase}, group phases $\theta_g = 0$, $g=1,\ldots, G$ that make the amplitudes $\alpha_g$ add up coherently, and $\alpha_g$ values maximized as described in Proposition \ref{prop:optCap} from the Takagi factorization of the rank-2 matrices ${\bf T}_g = \f_{d,g} \g_{a,g}^H + \left(\f_{d,g} \g_{a,g}^H \right)^T$.
\end{remark}

%%%%%%%%%%%%%%%%%%%%%%%%%%%%

\section{Simulation Results}
\label{sec:simBDRIS}
We consider a MIMO system aided by a fully-connected BD-RIS in which the Tx has coordinates (0, 0, 3) [m], the Rx is located at (200, 200, 1.5) [m], and the BD-RIS is at (20, 20, 20) [m]. The path loss is $PL = PL_0 -\beta 10 \log_{10} d$, where  $PL_{0} = -28$ dB is the path loss at a reference distance of $d_0 = 1$ meter and $\beta$ is the path loss exponent. We assume the direct MIMO channel, $\H_d$, is not blocked and has a rich multipath (Rayleigh) with a path loss exponent $\beta = 3.75$. The path loss exponent of the forward and backward BD-RIS channels is $\beta = 2$. Despite $\H_d$ not being fully blocked, the Tx-RIS-Rx has a Frobenius norm larger than that of the direct link, therefore, it contributes more significantly to the received signal power. The power spectral density for the additive noise is $\sigma^2 =-174 + 10\log_{10}B +F$ (dBm). The bandwidth is $B = 20$ MHz, and the noise factor $F= 10$ dB. 

In the first experiment, the MIMO system is $4 \times 4$, the transmitted power is $P_t=30$ dBm, and the forward and backward channels are pure LoS. Fig. \ref{fig:DeltaRatePt} shows the achievable rate for i) the proposed BD-RIS with optimal ${\bf R}_{xx}$ (Alg. 1); ii) the proposed BD-RIS with isotropic ${\bf R}_{xx}$; iii) the MIMO beamforming method in \cite{NeriniTWC2023}; and iv) a random BD-RIS with isotropic ${\bf R}_{xx}$. Fig. \ref{fig:DeltaRatePt} shows the BD-RIS designed for LoS channels with either optimized or isotropic Tx covariance matrix improves a MIMO beamforming scheme \cite{NeriniTWC2023}. This suggests that, even under a weak direct channel, a rank-1 Tx-RIS-Rx link may allow the transmission of a second stream over the equivalent channel.

\begin{figure}
    \centering
\includegraphics[width=.5\textwidth]{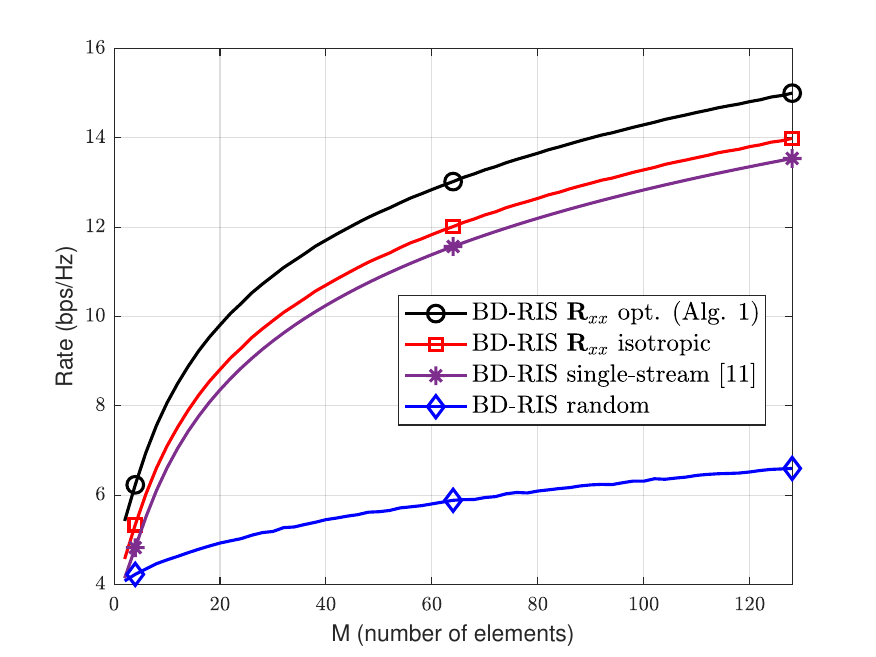}
     \caption{Achievable rate in a $4 \times 4$ MIMO system with pure LoS forward and backward channels for i) the proposed BD-RIS with optimal ${\bf R}_{xx}$ (Alg. 1); ii) the proposed BD-RIS with isotropic ${\bf R}_{xx}$; iii) the single-stream iterative method in \cite{NeriniTWC2023}; and iv) a random BD-RIS with isotropic ${\bf R}_{xx}$.}
	\label{fig:DeltaRatePt}
\end{figure}

In the second experiment, we evaluate the performance under Ricean forward and backward channels generated as
\begin{equation}
    \F = \sqrt{\frac{K}{1+K}} \F_{LoS} + \sqrt{\frac{1}{1+K}} \F_{NLoS},
\end{equation}
where $K$ is the Ricean factor that measures the relative strength between the direct, $ \F_{LoS}$ (modeled as a rank-1 pure LoS channel) component, and the scattered signal component, $\F_{NLoS}$ (modeled as a Rayleigh channel). The path loss exponent for the Ricean $\F$ and $\G$ channels is  $\beta = 2$, $P_t= 10 $ dBm, representing a moderate SNR situation, and the number of BD-RIS elements is $M=64$. The closed-form solutions for a BD-RIS proposed in \eqref{eq:optMIMOcap} and for an RIS are labeled as {\bf BD-RIS (LoS)} and {\bf RIS (LoS)} in the figure. In both cases, the transmit covariance matrix is obtained by solving a water-filling problem over the equivalent channel eigenmodes once the BD-RIS/RIS has been designed. As a comparison, we consider the following methods: 

\begin{itemize}
    \item {\bf BD-RIS (NLoS)}. This solution maximizes the capacity in a BD-RIS-assisted MIMO link in an NLoS scenario using the iterative algorithm in \cite{SantamariaSPAWC24}. 

    \item {\bf BD-RIS (non-rec. NLoS)}. The closed-form solution recently proposed in \cite{Emil2024Arxiv} for a scenario assisted by a passive but non-reciprocal (non-rec.) BD-RIS (i.e., a unitary but not symmetric BD-RIS matrix).

\item {\bf Random BD-RIS/RIS}. A random unitary and symmetric BD-RIS (or an RIS with random phases).

\item {\bf No RIS}. This is the scenario without RIS. The optimal transmit covariance matrix for the direct link $\H_d$ is used.
    
\end{itemize}

Fig. \ref{fig:R2x2M64P10bis} shows the achievable rates obtained by the different schemes for increasing values of the Ricean factor from $K=0$ (Rayleigh channels) to $K=10$ (channels with a dominant LoS component). For $K\geq1$, although the channels are far from being pure LoS, the closed-form BD-RIS solution is almost optimal, providing the same rate as the much more computationally expensive iterative solution of \cite{SantamariaSPAWC24}. The explanation is that, in this moderate-SNR scenario with only two antennas, the optimal transmission scheme is single-stream in many simulations. More significant differences are expected as the number of antennas and the SNR increase, or when the direct link is stronger. The BD-RIS solution outperforms the diagonal RIS, but both solutions tend to behave the same as $K$ increases. Finally, the solution proposed in \cite{Emil2024Arxiv} is always worse than the iterative design proposed in \cite{SantamariaSPAWC24}, being the BD-RIS solution in \cite{Emil2024Arxiv} less restrictive. A possible explanation is that the unitary solution in \cite{Emil2024Arxiv} assumes a blocked direct channel. This direct channel is not predominant in the considered scenario, but it influences the final result. 

\begin{figure}
    \centering
\includegraphics[width=.5\textwidth]{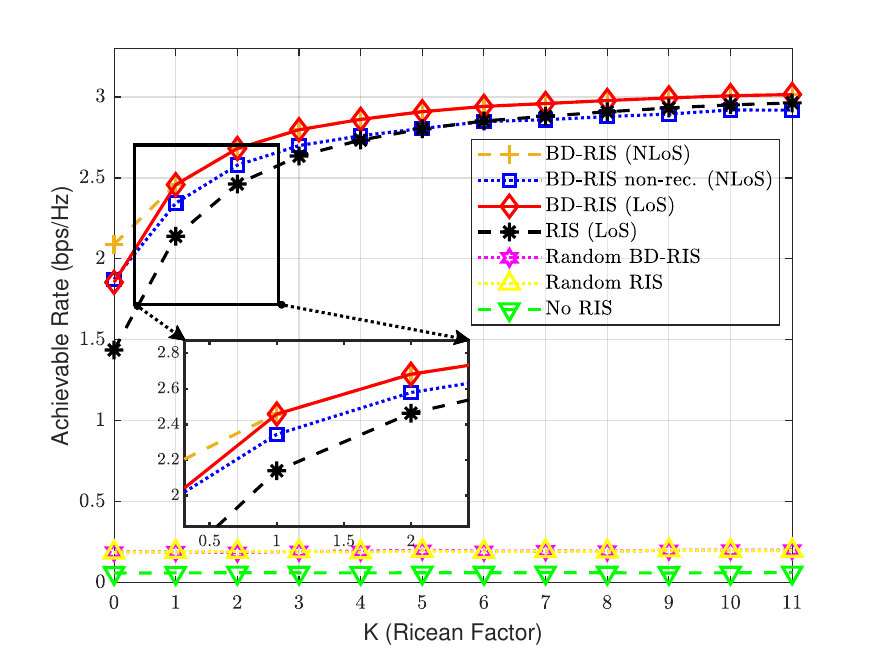}
     \caption{Achievable rate vs $K$ (Ricean factor for the BD-RIS channels) for several competing schemes. }
	\label{fig:R2x2M64P10bis}
\end{figure}

\section{Conclusion}
When both the forward and backward BD-RIS channels are LoS, we derived a closed-form solution for the rate maximization problem. The optimal solution has invariances that can be exploited to reduce the energy reflected by the BD-RIS, thus improving the energy efficiency in comparison to a diagonal RIS. Under forward and backward Ricean channels, the BD-RIS solution outperforms the diagonal RIS solution and it is competitive with iterative solutions when the direct MIMO link is weak. The extension of these results to multi-user scenarios is an interesting line of future work.

\section*{Appendix A: Proof of Proposition \ref{prop1new}}
\label{AppendixA}
Let $\B = \A + \alpha e^{j \theta} \f \g^H$. We want to show that
\begin{equation}
\label{eq:detexpression}
\det \left( \I + \B \B^H \right) = \det \left( \I + \A \A^H \right)  \left( 1 +   \,\Delta \right),
\end{equation}
with $\Delta$ defined as in Proposition \ref{prop1new}. Let us start by expanding
\begin{equation*}
    \B\B^H = \A\A^H + \f {\bf p}^H + \alpha e^{-j\theta}\A \g \f^H, 
\end{equation*}
where we have defined ${\bf p} = \alpha^2 \|\g\|^2 \f + \alpha e^{-j\theta} \A \g$. Defining $\C = \I + \A \A^H + \f {\bf p}^H$ and applying the matrix determinant lemma \cite[Sec. B.5.1]{Coherence} we get
\begin{equation}
\label{eq:expr1}
\det(\I+ \B \B^H ) = \det (\C) \left(1 + \, \alpha \, e^{-j\theta} \f^H \C^{-1} \A\g \right).
\end{equation}
A new application of the matrix determinant lemma, this time on the determinant of the matrix $\C$, gives us
\begin{equation}
\label{eq:matA}
   \det(\C) =  \det(\I + \A \A^H) \left( 1 + \, {\bf p}^H (\I + \A \A^H)^{-1} \f \right).
\end{equation}
Substituting ${\bf p} = \alpha^2 \|\g\|^2 \f + \alpha e^{-j\theta} \A\g$ in \eqref{eq:matA} we have
\begin{equation}
\label{eq:matAbis}
   \det(\C) =   \det(\I + \A \A^H) \left( 1 + \alpha^2 \|\g\|^2 \gamma_1 + \alpha e^{j\theta} \gamma_3   \right),
\end{equation}
where we have defined $\gamma_1 = \f^H\left( \I +\A \A^H \right)^{-1} \f$ and ${\gamma_3 = \g^H \A^H\left( \I + \A \A^H \right)^{-1} \f}$.

It remains to compute the second term in \eqref{eq:expr1}. First, we apply the matrix inversion lemma \cite[Sec. B.4.2]{Coherence} to compute $\C^{-1}$
\begin{equation}
    \C^{-1} = \E^{-1} - \frac{\E^{-1}\f {\bf p}^H \E^{-1}}{\left(1 + {\bf p}^H \E^{-1}\f \right)},
\end{equation}
where we have defined $\E = \I + \A \A^H$. Then, we have
\begin{multline}
   \left(1 +\, \alpha \, e^{-j\theta} \f^H \C^{-1} \A \g \right) = 1 +  \, \alpha \, e^{-j\theta} \gamma_3^*  \\
   -  \, \alpha \, e^{-j\theta} \frac{\gamma_1 (\alpha^2 \|\g\|^2 \gamma_3^* + \alpha e^{j \theta }\gamma_2)}{\left(1 + {\bf p}^H \E^{-1}\f \right)}, \label{eq:expre2}
\end{multline}
where we have defined $\gamma_2 = \A^H \g^H (\I + \A \A^H)^{-1} \A \g$. Substituting \eqref{eq:expre2} into \eqref{eq:expr1} we finally obtain \eqref{eq:detexpression}.

\addcontentsline{toc}{section}{Appendices}
\renewcommand{\thesubsection}{\Alph{subsection}}
%\subimport{./}{appendix1.tex}

\balance
%\newpage
\bibliographystyle{IEEEtran}
\bibliography{Main}
\end{document}